\documentclass[sigconf,nonacm]{aamas}

\usepackage{balance} 
\usepackage{algorithm}
\usepackage{algorithmic}
\usepackage[table]{xcolor}
\usepackage{colortbl}
\usepackage{amsthm}
\usepackage{amsmath}
\usepackage{graphicx}
\usepackage{subcaption}
\usepackage{xparse} 
\usepackage{xspace} 
\usepackage{thmtools}
\usepackage{thm-restate}
\usepackage{cleveref}
\usepackage{xcolor}
\usepackage{soul}
\usepackage[normalem]{ulem} 
\usepackage{cancel}

\newif\ifhighlight
\highlightfalse   

\doi{LGSS5881}

\makeatletter
\gdef\@copyrightpermission{
  \begin{minipage}{0.2\columnwidth}
   
  \end{minipage}\hfill
  \begin{minipage}{0.8\columnwidth}
 
  \end{minipage}
  \vspace{5pt}
}
\makeatother

\setcopyright{ifaamas}
\acmConference[AAMAS '26]{Proc.\@ of the 25th International Conference
on Autonomous Agents and Multiagent Systems (AAMAS 2026)}{May 25 -- 29, 2026}
{Paphos, Cyprus}{C.~Amato, L.~Dennis, V.~Mascardi, J.~Thangarajah (eds.)}
\copyrightyear{2026}
\acmYear{2026}
\acmDOI{}
\acmPrice{}
\acmISBN{}

\acmSubmissionID{13}

\title{Maximin Share Guarantees via Limited Cost-Sensitive Sharing}

\author{Hana Salavcova}
\affiliation{
  \institution{Charles University}
  \city{Prague}
  \country{Czech Republic}}
\email{salavcovah@gmail.com}

\author{Martin Černý}
\affiliation{
  \institution{Charles University}
  \city{Prague}
  \country{Czech Republic}}
\email{cerny@kam.mff.cuni.cz}

\author{Arpita Biswas}
\affiliation{
  \institution{Rutgers University}
  \city{New Jersey}
  \country{United States}}
\email{a.biswas@rutgers.edu}

\begin{abstract}
We study the problem of fairly allocating indivisible goods when limited sharing is allowed, that is, each good may be allocated to up to $k$ agents, while incurring a cost for sharing. While classic maximin share (MMS) allocations may not exist in many instances, we demonstrate that allowing controlled sharing can restore fairness guarantees that are otherwise unattainable in certain scenarios. {(1)~Our first contribution shows that exact maximin share (MMS) allocations are guaranteed to exist whenever goods are allowed to be cost-sensitively shared among at least half of the agents and the number of agents is even; for odd numbers of agents, we obtain a slightly weaker MMS guarantee. (2)}~{We further design a Shared Bag-Filling Algorithm that guarantees a $(1 - C)(k - 1)$-approximate MMS allocation, where $C$ is the maximum cost of sharing a good. Notably, when $(1 - C)(k - 1) \geq 1$, our algorithm recovers an exact MMS allocation.} {(3)}~We additionally {introduce} the \textit{Sharing Maximin Share} (SMMS) fairness notion, a natural extension of MMS to the $k$-sharing setting. {(4)~We show that SMMS allocations always exist under identical utilities and for instances with two agents. (5)~We construct a counterexample to show the impossibility of the universal existence of an SMMS allocation. (6)}~Finally, we establish a connection between SMMS and constrained MMS (CMMS), yielding approximation guarantees for SMMS via existing CMMS results. {These contributions provide deep theoretical insights for the problem of fair resource allocation when a limited sharing of resources are allowed in multi-agent environments.}
\end{abstract}


\newtheorem{observation}{Observation}

\NewDocumentCommand{\removeandswap}{ O{A} O{g} O{i} O{j} O{} }{%
  \IfNoValueTF{#5}
    {\ensuremath{(#1 \setminus^{#5} #2)_{#3 \leftrightarrow #4}}}%
    {\ensuremath{(#1 \setminus^{#4} #2)_{#3 \leftrightarrow #4}}}%
}

\newtheorem*{repproposition}{Proposition \ref{prop:smms-for-two-agents}}

\newtheorem*{repproposition2}{Proposition \ref{prop:smms-equiv-cmms}}

\newtheorem*{repproposition3}{Proposition \ref{prop:smms-for-identical-agents}}

\NewDocumentCommand{\MMS}{ O{i} O{n} O{M} }{%
  \ensuremath{\operatorname{MMS}_{#1}^{#2}(#3)}\xspace%
}

\NewDocumentCommand{\MMSi}{ O{i} }{%
  \ensuremath{\operatorname{MMS}_{#1}}\xspace%
}

\NewDocumentCommand{\kMMS}{ O{i} O{k} O{n} O{M} }{%
  \ensuremath{\operatorname{MMS}_{#1}^{#3,#2}(#4)}\xspace%
}

\NewDocumentCommand{\SMMS}{ O{i} O{k} O{n} O{M} }{%
  \ensuremath{\operatorname{SMMS}_{#1}^{#3,#2}(#4)}\xspace%
}

\NewDocumentCommand{\SMMSc}{ O{C} O{i} }{%
  \ensuremath{\operatorname{SMMS}^{#1}_{#2}}\xspace%
}

\NewDocumentCommand{\fullSMMSc}{ O{C} O{i} }{%
  \ensuremath{\operatorname{\overline{SMMS}}^{#1}_{#2}}\xspace%
}

\newcommand{\cellcolorval}[1]{%
  \ifdim #1 pt = 1pt
    \cellcolor{green!50}%
  \else\ifdim #1 pt > 0.89pt
    \cellcolor{green!30}%
  \else\ifdim #1 pt > 0.75pt
    \cellcolor{green!20}%
  \else\ifdim #1 pt > 0.59pt
    \cellcolor{yellow!45}%
  \else\ifdim #1 pt > 0.49pt
    \cellcolor{yellow!30}%
  \else\ifdim #1 pt > 0.35pt
    \cellcolor{red!20}%
  \else\ifdim #1 pt > 0.19pt
    \cellcolor{red!30}%
  \else
    \cellcolor{red!40}%
  \fi\fi\fi\fi\fi\fi\fi
  #1%
}

\begin{document}


\pagestyle{fancy}
\fancyhead{}


\maketitle 


\section{Introduction}

The theory of fair allocation of indivisible items has been receiving significant attention, driven both by theoretical interest and by a wide range of practical applications~\cite{walshfair,Aziz2022,Amanatidis2023}. A substantial body of work has explored various generalizations, relaxations, and constraints that arise in real-world allocation problems. For instance, several studies have addressed settings where not all items must be allocated, leading to models where certain fairness guarantees exist, which is otherwise unattainable~\cite{Caragiannis2019,Chaudhury2021b,Chaudhury2021,Berger2022}. 
These extensions reflect a growing recognition of the need for fair allocation settings and methods that are both flexible and applicable to complex, structured domains.

One of the most prominent fairness notions in the area of fair division of indivisible goods is \emph{maximin share} (MMS) fairness, where each agent is guaranteed a bundle at least as much as their maximin share value. Informally, an agent’s maximin share value represents the maximum utility they can guarantee themselves by partitioning the goods into as many bundles as there are agents, under the assumption that they receive the least valuable bundle. This concept captures a natural and compelling benchmark for fairness, as it guarantees what an agent could ensure for themselves in a worst-case, divide-and-choose scenario. MMS was formally introduced in~\cite{Budish2011} in the context of course allocation. 

Despite its intuitive appeal, MMS allocations are not guaranteed to exist in general. In fact, counterexamples with as few as three agents and a small number of goods have shown that no MMS exists~\cite{Kurokawa2018,Feige2021}. This inherent impossibility has led researchers to focus on approximate MMS allocations, where each agent receives at least a fraction of their MMS value. The study of such approximations has become a central direction in the literature, yielding progressively stronger guarantees over time~\cite{Barman2020,Ghodsi2021,Garg2021,Feige2021} with the best known approximation factor of $\left(\frac{3}{4} + \frac{3}{3836}\right)$ under additive valuations established by~\citet{Akrami2024}.

Much of the prior work on fair allocation assumes that the set of items must be partitioned disjointly among the set of agents, meaning that no good can be shared among multiple agents.  However, many practical scenarios, such as allocation of shared computing access~\cite{jacquet2024sweetspotvm} and community energy storage systems~\cite{chang2022shared}, may allow the sharing of resources in a structured way. {For instance, consider allocating access to high-demand laboratory equipment at a university, such as a high-resolution electron microscope, a specialized DNA sequencer, and a powerful computing server. These devices are expensive, scarce, and cannot be physically subdivided, but each can be shared among a limited number of users via scheduled access time slots or usage quotas. However, sharing may come at a cost, and the effective value of a resource to each user can diminish as more users share it. A purely non-sharing allocation may leave some researchers without access to any resource, failing to meet basic fairness expectations. In contrast, allowing resources to be shared in a structured, bounded manner can help achieve fairness even when sharing incurs a cost, motivating the study of fair allocation under bounded, cost-sensitive sharing.} {Allowing such structured sharing can potentially lead to allocations with improved fairness guarantees.}{This leads to the central question:}

\begin{center}
\emph{Can a fair allocation be obtained by relaxing the classical no-sharing constraint in structured ways, particularly in settings where fairness is otherwise unattainable?}
\end{center}

In this work, we allow each good to be allocated to at most~$k$ agents, while explicitly accounting for the costs incurred due to sharing. This framework significantly broadens the applicability of fair allocation, making it more aligned with real-world scenarios where limited sharing is often inevitable or desirable. A consequence of the result by~\citet{Akrami2025} and~\citet{Barman2025} implies that sharing among two agents is sufficient to ensure MMS fairness, assuming no cost is incurred when sharing goods. Including the cost of sharing makes the model more realistic, as sharing typically reduces the benefit that a good provides to each user. 
Therefore, we formalize the $k$-sharing fair allocation problem and introduce models to capture the cost of sharing {(in Section~\ref{sec:problem})}.  Our main theoretical contributions are as follows:

\begin{enumerate}    
    \item {In Section~\ref{sec:mms}, }we show that exact maximin share (MMS) allocations can be achieved when goods are allowed to be cost-sensitively shared among at least half of the agents and when there is an even number of agents, under certain cost models. When agents are odd in number, we establish a slightly weaker MMS guarantee {(Theorem~\ref{prop:mms-for-2k-generous})}. Further, for any cost-sharing model, we provide {Shared Bag-Filling Algorithm} for computing approximate MMS allocations in the shared setting {(Theorem~\ref{prop:Bag-Filling})}.
    
    \item {In Section~\ref{sec:smms}}, we introduce a stronger fairness notion for $k$-sharing allocations, \emph{Sharing Maximin Share (SMMS)}. We show that SMMS always exists when there are only two agents  ({Proposition~\ref{prop:smms-for-two-agents}}) and under identical utilities ({Proposition~\ref{prop:smms-for-identical-agents}}). Interestingly, the existing MMS counterexamples \cite{Kurokawa2018,Feige2021} admit feasible SMMS allocations, prompting us to construct a new counterexample to disprove universal SMMS existence
    {(Theorem~\ref{prop:smms-dont-exist})}.

    \item  Additionally, we establish a connection between SMMS for shared settings and constrained MMS (CMMS), yielding approximation guarantees for SMMS {(Proposition~\ref{prop:approx-MMS-SMMS})}.
\end{enumerate}

\section{Related Work}
The literature on discrete fair allocation is extensive and rapidly evolving, typically operating under the two key assumptions: \textit{non-shareability}, where an agent's allocation is disjoint from others, and \emph{completeness}, where all items must be allocated among all the agents. Recent survey articles, such as those by \citet{walshfair}, \citet{Aziz2022}, and \citet{Amanatidis2023}, offer comprehensive insights into the algorithmic and complexity dimensions of the area of fair allocation. In addition to these broad overviews, more focused surveys have explored specific subdomains: for instance, \citet{aleksandrov2020online} examined fairness in dynamic resource allocation, while \citet{Suksompong2021} and \citet{Biswas2023} investigate fairness under various structural set constraints.

Despite significant advances, the well-studied maximin share (MMS)~\cite{Budish2011} fairness guarantee does not hold under the standard assumptions of non-shareability and completeness~\cite{kurokawa2016can,Kurokawa2018,Feige2021}. Recent research has made progress on achieving various fairness guarantees by relaxing these classical assumptions through innovative approaches: (1)~allowing some items to remain unallocated (as \emph{charity})~\cite{Chaudhury2021,biswas2023algorithmic}, (2)~treating items expandable such as school seats or course seats, and allowing their capacities to be increased~\cite{procaccia2024school,santhini2024approximation}, {(3) considering some items to be divisible~\cite{Bismuth2024,Sadomirskiy2022}} and ({4})~creating duplicate copies of items to enable multi-allocation~\cite{Barman2025,Akrami2025}. \citet{Barman2025} show that MMS fairness can be achieved under additive valuations by duplicating each item at most once, and Akrami et al.~\cite{Akrami2025} show that MMS can be achieved with one duplicate copy of at most $\lfloor n/2\rfloor$ items. While these results imply that allowing sharing among two agents helps achieving MMS fairness, the results fall short when the cost of sharing is considered. To the best of our knowledge, no prior work has considered cost-sensitive shared fair allocation setting. 
\section{Preliminaries}

An instance of a discrete fair allocation problem is typically represented by a triple $I = (N, M, v)$ 
where $N$ is a set of $n$ agents, $M$ is a set of $m$ indivisible items, and $v = (v_1, \dots, v_n)$ is a valuation profile, where each $v_i: 2^M \rightarrow \mathbb{R}$ is a valuation function for agent $i$ assigning a value to every subset of items $S\subseteq M$. In this work, we assume non-negative and additive valuations, that is, $v_i(S) \in \mathbb{R}_{0}^+$ for all $S\subseteq M$ and for each agent $i\in N$, and $v_i(S) = \sum_{g \in S}v_i(\{g\})$. Since we assume non-negative valuations, we will refer to items as \emph{goods}. For brevity, we will write $v_i(g)$ instead of $v_i(\{g\})$ throughout the paper.

In the classical setting, an allocation is a tuple $A = (A_1, A_2, \dots, A_n)$ where each $A_i \subseteq M$ is the bundle allocated to an agent $i$ and satisfy two properties, namely (1)~\emph{non-shareability}: $A_i \cap A_j = \emptyset$ for all $i \ne j$
and (2)~\emph{completeness}: $\bigcup_{i \in N} A_i = M$. In other words, an allocation corresponds to a partition of the set of goods $M$ into $n$ disjoint subsets, with no good assigned to more than one agent.

In this work, we focus on \emph{maximin share} (MMS) fairness~\cite{Budish2011}, which is one of the most well-studied notions in discrete fair allocation literature. It guarantees each agent at least their \emph{maximin share value}~($\MMS[i]$), that is, $v_i(A_i) \geq \MMS[i]$ where
\begin{equation}\label{def:MMS}
    \MMS[i] := \max_{(A_1, \dots, A_n) \in \mathcal{A}^n(M)} \min_{j \in [n]} v_i(A_j),
\end{equation}
and $\mathcal{A}^n(M)$ denotes the set of all possible $n$-partitions of the set of goods~$M$. In other words, the \emph{maximin share value}~$\MMSi$ for an agent~$i$ is the maximum value she can guarantee for herself by partitioning the goods into $n$ bundles and receiving the least valued bundle.


While MMS often fails to exist in discrete fair allocation problems under the standard assumptions of non-shareability, we reveal new pathways to fairness by introducing cost-sensitive restricted sharing. We now formalize the problem setting in the subsequent section.

\section{Problem Formulation}\label{sec:problem}

We investigate fair allocation settings where each good is allowed to be shared among a fixed number of agents $k$. 
We introduce \emph{$k$-sharing allocations}.
\begin{definition}
\label{def:k-sharing-allocation}
A \emph{$k$-sharing allocation} is a tuple $A = (A_1,\dots,A_n)$ where the bundles $A_i \subseteq M$ satisfy two properties, namely (1)~\emph{$k$-limited shareability}: $\bigcap_{i \in S} A_i = \emptyset$, for every group $S \subseteq N$ of size $|S|>k$ and (2)~\emph{completeness}: $\bigcup_{i \in N} A_i = M$. Further, $\mathcal{A}_k^n(M)$ denotes the set of all possible $k$-sharing allocations of goods $M$.
\end{definition}
\noindent
The \emph{$k$-limited shareability} condition ensures that each good can be shared among at most $k$ agents. To denote the set of agents sharing a good \( g \) under allocation \( A \), we use 
\[
    N_g(A) := \{ i \in N : g \in A_i \}.
\]
 According to the \emph{$k$-limited shareability} condition, we are required to satisfy \( |N_g(A)| \leq k \).  
Additionally, we say that a \( k \)-sharing allocation is \emph{fully-shared} if every good is shared by exactly \( k \) agents, that is, \( |N_g(A)| = k \) for all \( g \in M \).

In the classical ($1$-sharing) model, goods are allocated exclusively, and each agent’s utility is determined solely by the goods in her bundle. However, in the \( k \)-sharing setting, the utility an agent derives from a good may depend not only on having access to it but also on how many share it. Thus, rather than defining utility based only on the agent's individual bundle, we consider utility functions \( u_i \) that take the full allocation \( A \) as input.

\begin{definition}
The \emph{utility} of an agent $i$ in a k-sharing setting,  $u_i \colon \mathcal{A}_k^n(M) \to \mathbb{R}_0^+$, is defined as
\begin{equation}\label{eq:utility}
    u_i(A) = \sum_{g \in A_i} \big[1-c_{i,g}(N_g(A))\big] \cdot v_i(g)
\end{equation}
with $c_{i,g}(N_g(A)) \in [0,1]$ denoting the cost incurred by agent $i$ for sharing good $g$ with agents $N_g(A)$. Additionally, 
\begin{equation}\label{eq:zero-costs}
    c_{i,g}(N_g(A)) = 0\text{ if }|N_g(A)| = 1.
\end{equation}
\end{definition}





In this work, we consider \emph{goods-based} cost models, where for every agent $i \in N$ and good $g \in M$, the cost depends only on the number of agents sharing $g$, and is given by $c_{i,g}(N_g(A)) = c_g(|N_g(A)|)\in[0,1]$.

\begin{enumerate}
\item \textbf{Cost-free sharing:} $c_{g}(|N_g(A)|) = 0$ for every $g \in M$ (no cost of sharing). Here, the utility of agent \( i \) equals the valuation of their allocated bundle (including their shared goods), \( u_i(A) = v_i(A_i) \).

\item \textbf{Equal-share cost-sharing:} $c_{g}(|N_g(A)|) = 1 - \frac{1}{|N_g(A)|}$ for all $g \in M$. Therefore, each agent receives a utility of $\frac{1}{
|N_g(A)|}$ fraction of her valuation for the good $g$ allocated to her,
\(
    u_i(A) = \sum_{g \in A_i} \frac{1}{|N_g(A)|} \cdot v_i(g).
\)

\item \textbf{Generous cost-sharing:} \( c_{g}(|N_g(A)|) \in [0,1-\frac{1}{|N_g(A)|}] \) for every $g \in M$. 
Further, we assume costs are non-decreasing in the number of sharers: \( c_g(\ell') \leq c_g(\ell) \) for \( \ell' < \ell \). 
When $c_{g}(|N_g(A)|)=0$ for all $g\in M$, the model becomes \emph{cost-free}.  When $c_{g}(|N_g(A)|)=1-\frac{1}{|N_g(A)|}$ for all $g\in M$, it represents \emph{equal-share}.  


\end{enumerate}

Using all the components, we now denote a cost-sensitive $k$-sharing fair allocation instance as a tuple $(N,M,k,v,c)$ where $N$ is the set of $n$ agents, $M$ is the set of $m$ goods, $k$ is the sharing constraint, $v=(v_i,\ldots,v_n)$ is the valuation profile, and the sharing-cost function is denoted by $c=(c_{g}(\ell)))_{g \in M, \ell\in \{1,\ldots,k\}}$. 

\section{MMS under Cost-Sensitive Sharing}\label{sec:mms}


A $k$-sharing allocation \( A \in \mathcal{A}_k^n(M) \) satisfies MMS if \( u_i(A) \geq \MMS \) for all \( i \in N \) (as defined in Equation~\ref{def:MMS}). While sharing increases the flexibility to achieve fairer outcomes—by allowing agents to receive more (potentially shared) goods—it also introduces costs that complicate the structure of the utility function. For any given good, the utility an agent derives can be, due to the costs, lower when the good is shared than when it is received exclusively.



We now highlight the generalizability of the equal-share cost-sharing model in terms of satisfying MMS, or any other threshold-based fairness notions.

\begin{observation}\label{obs:equal_share}
    In a k-sharing setting, if an allocation satisfies MMS under equal-share cost-sharing model, then the allocation also satisfies MMS under any \emph{generous} cost-sharing model, including cost-free sharing model. 
\end{observation}

Observation~\ref{obs:equal_share} follows directly from the definitions of the cost-sharing models, implying that an agent's utility for a given $k$-sharing allocation under the \emph{equal-share} model $c_{g}(|N_g(A)|)=1-\frac{1}{|N_g(A)|}$ increases when we consider lesser cost of sharing,  $0\leq c_{g}(|N_g(A)|)<1-\frac{1}{|N_g(A)|}$. 


Therefore, we focus on the equal-share cost-sharing model to investigate the existence of MMS in the cost-sensitive \(k\)-sharing setting. Intuitively, allowing goods to be shared among more agents increases the likelihood of achieving a fair allocation. Indeed, when \(k = n\), each item can be shared among all agents, trivially satisfying MMS. 

\subsection{Sharing among at Least Half of the Agents}
Assuming $k\geq n/2$ under equal-share cost-sharing model, we establish the existence of MMS when $n$ is even, with a weaker MMS guarantee when $n$ is odd.

\begin{theorem}\label{prop:mms-for-2k-generous}
    In every $k$-sharing instance $(N, M, k, v, c)$ under equal-share cost-sharing model, with $k \geq n/2$ agents, there exists an allocation $A$ that for every agent $i$ satisfies:
    \begin{itemize}
        \item $u_i(A) \geq \MMS[i][n]$ if $n$ is even,
        \item $u_i(A) \geq \MMS[i][n+1]$ if $n$ is odd.
    \end{itemize}
\end{theorem}

\begin{proof} We provide a constructive proof for both the cases. First, we consider $n$ to be even. 
We arrange all the agents into \(\ell = n/2\) disjoint pairs. For each pair of agents, \(\{i, j\}\), we define a classical (no-sharing) fair allocation instance, $I_{i,j}=(N =\{i,j\}, M, v=(v_i, v_j))$. Since MMS allocation can always be obtained when there are 2 players, we find one such allocation and denote it as $A_{i,j}=(A_i, A_j)$. Note that the allocated bundles between each pair are disjoint, $A_i\cap A_j = \emptyset$. Now, let us consider all such allocated bundles from all $\ell$ pairs of agents $A=(A_1,\ldots,A_n)$. Since the allocated bundles between each of the $\ell$ pairs were disjoint, each good $g$ appears in exactly $\ell$ different bundles in $A$, and the utility is reduced by {a factor of} $1/\ell$. Therefore, for any agent $i$, the utility of $A_i$ upon sharing each good with $\ell$ agents, under equal-share cost-sharing model is:
\begin{align}
u_i(A) 
&= \sum_{g \in A_i} (1 - c_g(\ell)) \cdot v_i(g)\notag\\
& = \sum_{g \in A_i} \frac{1}{\ell} \cdot v_i(g) = \frac{1}{\ell} \cdot v_i(A_i)\notag \\
& \geq \frac{1}{\ell}\max_{B \in \mathcal{A}_1^2(M)}\min_{h \in \{i,j\}}v_i(B_h) = \frac{1}{\ell} \cdot \mbox{MMS}^2_i(M)\label{eq:u1}
\end{align}





    Now, let \( B^* \in \mathcal{A}_1^{2\ell}(M) \) be the $\MMS[i]$ maximizer. Let \(C\) be the union of the \(\ell\) least-valued bundles in \(B^*\) according to agent~\(i\), and let \(C' = M \setminus C\). Both \(C\) and \(C'\) have total value at least \( \ell \cdot \min_{h \in N} v_i(B^*_h) \). Therefore, the bipartition \( (C, C') \in \mathcal{A}_1^2(M) \) satisfies:
    \begin{equation}
    \min\{v_i(C), v_i(C')\}\geq \ell \cdot \MMS[i]\label{eq:mms2}.
    \end{equation}
    
   Inequality (\ref{eq:mms2}) along with the fact that $\mbox{MMS}^2_i(M)\geq \min \{v_i(D), v_i(D')\}$ for any $(D,D')\in \mathcal{A}^2_1(M)$, and combining with Inequality~(\ref{eq:u1}), we obtain $u_i(A) \geq \MMS[i]$ when $n$ is even.



    Next, we consider $n$ to be odd. In this case, we add a \emph{dummy} agent with zero valuation for all goods. With $n+1$ (even) agents, we apply the construction for even agents, obtaining $k$-sharing allocation of $n+1$ agents, $A$, which satisfies: $u_i(A) \geq \MMS[i][n+1]$ for all agents $i \in N$. Finally, the \emph{dummy} agent is removed and its goods are arbitrarily distributed among all the agents. Since this can only increase the utility of each agent in $N$, the result follows.
\end{proof}

Next, we show how to compute an $\alpha$-approximate MMS allocation in polynomial time for any $k \geq 2$, possibly yielding a better approximation than algorithms without sharing.


\subsection{Approximate MMS via Cost-Sensitive k-Sharing}
Our polynomial-time algorithm, {Shared Bag-Filling}, (Algorithm~\ref{alg:modified-bag-filling}) works {under any \emph{general} cost model with $c_{i,g}(N_g(A)) \in [0,1]$, including the generous cost-sharing models}.{for \textit{any} cost-sharing model, not only for generous goods-based ones. This formally means that the costs satisfy $c_{i,g}(N_g(A)) \in [0,1]$ without any additional restrictions.}
{The algorithm} computes an $\alpha$-approximate MMS allocation, with the guarantee depending on the maximum sharing cost:
\[
C := \max_{i \in N, g \in M, S \subseteq N, |S| \leq k} c_{i,g}(S).
\]
Notably, Algorithm~\ref{alg:modified-bag-filling} requires neither the knowledge nor the computation of $C$.



\begin{algorithm}[t]
\caption{
{Shared Bag-Filling Algorithm}
}
\label{alg:modified-bag-filling}
\begin{algorithmic}[1]
\REQUIRE $I=(N, M, k, v, c)$
\ENSURE The output allocation $A$ is $\alpha$-MMS where $\alpha=\min\{1,(1-C)(k-1)\}$.

\STATE $\tilde{N} \gets N$, $\tilde{M} \gets M$  \hfill // remaining agents and goods

\nonumber{$\ $}\\
\nonumber{\textbf{Phase 1}}

\WHILE{$\exists i \in N, \exists g \in M$ s.t. $v_i(g) \geq \frac{v_i(\tilde{M})}{|\tilde{N}|}$}
    \STATE $A_i \gets \{g\}$; $\quad\tilde{M} \gets \tilde{M} \setminus \{g\}$; $\quad\tilde{N} \gets \tilde{N} \setminus \{i\}$
\ENDWHILE

\nonumber{$\ $}\\
\nonumber{\textbf{Phase 2}}
\STATE $v_i^\text{norm}(g) \gets \frac{|\tilde{N}|}{v_i(\tilde{M})} \cdot v_i(g)$ for all $i \in \tilde{N}$ and all $g \in \tilde{M}$
\STATE Let $\mathcal{M}$ be the multiset of $k$ copies of each good from $\tilde{M}$
\hfill // denote the multiplicity of $g$ in $\mathcal{M}$ by $\xi_{\mathcal{M}}(g)$
\STATE $\tilde{v}_i(g)=\frac{1}{k} \cdot v_i^\text{norm}(g)$ for all $i \in \tilde{N}, g \in \tilde{M}$

\WHILE{$|\tilde{N}| > 1$}
    \STATE Initialize empty bag $B \gets \emptyset$
    \WHILE{exists $g \in \mathcal{M}$ s.t. $\xi_{\mathcal{M}}(g) = |\tilde{N}|$}
        \STATE $B \gets B \cup \{g\}$
        \STATE $\mathcal{M} \gets \mathcal{M} \setminus \{g\}$ \hfill // ensures $\xi_{\mathcal{M}}(g) < |\tilde{N}|$
    \ENDWHILE
    \WHILE{$\tilde{v}_i(B) < \frac{k-1}{k}$ for all $i \in N$}
        \STATE Select arbitrary $g \in \mathcal{M}$ such that $g \notin B$
        \STATE $B \gets B \cup \{g\}$ and $\mathcal{M} \gets \mathcal{M} \setminus \{g\}$
    \ENDWHILE
    \STATE Choose any agent $i^* \in \tilde{N}$ such that $\tilde{v}_{i^*}(B) \geq \frac{k-1}{k}$
    \STATE $A_{i^*} \gets B$
    \STATE $\tilde{N} \gets \tilde{N} \setminus \{i^*\}$
    \hfill // invariant $\xi_{\mathcal{M}}(g) \leq |\tilde{N}|$ holds
\ENDWHILE

\STATE $A_j \gets \mathcal{M}$, where $j$ is the remaining agent in $\tilde{N}$

\RETURN $A = (A_1, \dots, A_n)$
\end{algorithmic}
\end{algorithm}


Algorithm~\ref{alg:modified-bag-filling} builds on the classic bag-filling technique~\cite{Ghodsi2021,garg2019approximating,Garg2021}, adapted to the $k$-sharing setting. It runs in two phases. 

Phase~1 allocates \emph{large goods}, i.e., those with $v_i(g) \geq {v_i(M)/|N|}$ ${v_i(\tilde{M})/|\tilde{N}|}$ for some $i \in {\tilde{N}}$. Since $g$ already meets agent~$i$’s MMS guarantee {(in the instance with $\tilde{M}$, $\tilde{N}$)}, we can assign it directly to her and exclude her from further allocation, as by the monotonicity of MMS (see Appendix~\ref{app:bag-filling}), this 
does not reduce the MMS values of the remaining agents on the remaining goods. 

Phase~2 begins by normalizing the valuations of the remaining agents $N'$ by setting $v_i^\text{norm}(M') = |N'|$, where $M'$ are the remaining goods. This guarantees $\operatorname{MMS}_i^\text{norm} \leq 1$. In the subsequent steps, $u_i^\text{norm}(A) \geq \MMSi^\text{norm}$ is ensured for each remaining agent, which implies the original MMS guarantee. Details are deferred to Appendix~\ref{app:bag-filling}. 

Each good $g \in M’$ is then replaced by $k$ virtual copies (shares), and valuations are scaled so that each share has value $\tilde{v}_i(g) = \frac{1}{k} \cdot v_i^\text{norm}(g)$. This transformation preserves total value and ensures that the final outcome corresponds to a fully-shared allocation over $M’$, while reducing the problem to allocating small items.

Phase~2 then proceeds similarly to classic bag-filling, with the constraint that no bag contains multiple shares of the same good. It starts with an empty bag $B$, adds one share of each good whose number of remaining shares equals the number of remaining agents, then continues adding shares arbitrarily until some agent $i$ satisfies $\tilde{v}_i(B) \geq \frac{k-1}{k}$. The bag is allocated to agent~$i$, and the process repeats. The last agent receives all remaining goods. For first $|N'|-1$ agents,

\begin{equation}\label{eq:what-player-receives}
\begin{split}
    u_i^\text{norm}(A) &\geq (1 - C) \cdot v_i^\text{norm}(A_i) = (1 - C) \cdot k \cdot \tilde{v}_i(A_i) \\
         &\geq (1 - C) \cdot k \cdot \frac{k-1}{k} = (1-C)\cdot (k-1).
\end{split}
\end{equation}
Since the valuations are normalized, we have $\MMSi^\text{norm}\leq 1$. Therefore, it holds $u_i^\text{norm}(A) \geq (1 - C)(k - 1)\cdot\MMSi^\text{norm}$. 

There are two issues to address in the analysis of Phase~2. First, it is not immediately clear that the algorithm cannot get stuck---i.e., that at any point, there is always an item that can be added to the current bag when it is not acceptable to every agent{s}. Second, it's not obvious that the last agent, who receives the remaining goods, values them at least at their $(1 - C)(k - 1)$-MMS value.

{The key to understanding why these concerns do not arise lies in the following lemma, which shows that at every iteration of the algorithm there remain shares of sufficient amount and value to satisfy the remaining agents.}


\begin{lemma}
\label{lem:invariant}
At {the beginning and the end of} every iteration of Phase 2 {(Lines 8-21)} of Algorithm~\ref{alg:modified-bag-filling}, for every $i \in \tilde{N}$, it holds
\[
    \tilde{v}_i(\mathcal{M}) \geq |\tilde{N}|,
\]
where $\mathcal{M}$ {is} the multiset of $k$ copies of each good from $\tilde{M}$.
\end{lemma}
\begin{proof}
Initially, $\tilde{v}_i(\mathcal{M}) = k  \cdot \tilde{v}_i(M') = k \cdot \frac{1}{k} \cdot v_i^\text{norm}(M') = |N'|$. Assume that the invariant holds at the beginning of some iteration, when the remaining shares are $\mathcal{M}$.

{If the while loop at Line 14 does not execute, that is, already the initial bag $B$ (containing exactly one share of each good whose number of remaining shares equals the number of remaining agents) satisfies $\tilde{v}_{i^*}(B) \ge \frac{k-1}{k}$ for some agent~$i^*$, then the bag is immediately assigned to agent~$i^*$. In this case,}
the total value of $B$ is at most $\frac{1}{|\tilde{N}|}\tilde{v}_i(\mathcal{M})$ for every agent. This means
\begin{align*}
\tilde{v}_i(\mathcal{M} \setminus B) 
&\geq \tilde{v}_i(\mathcal{M}) - \frac{1}{|\tilde{N}|} \tilde{v}_i(\mathcal{M}) 
= \frac{|\tilde{N}| - 1}{|\tilde{N}|} \tilde{v}_i(\mathcal{M})
\end{align*}
which is by the assumption larger or equal to $|\tilde{N}|-1$.

{Otherwise, if the loop starting at Line 14 executes at least once,}
additional goods are added to the bag $B$ according to Lines~15-18. During this process, before the last good was added, the value of $B$ was at most $(k-1)/k$ for every agent. As all shares are valued by all agents at most $1/k$, the value of the bag is at most $1$. In both cases, the invariant is preserved.
\end{proof}

{Combining Lemma~\ref{lem:invariant} with the Phase~2 invariant—that the number of shares of every remaining item is at most the number of remaining agents—we conclude that a dead-end state cannot occur because if there is no share left to add to the bag, then every remaining agent must already prefer the bag.}

{\begin{proposition}\label{prop:termination}
    Phase 2 does not enter a dead-end state.
\end{proposition}
\begin{proof}
    Suppose that at some point, no share can be added to the current bag $B$ without violating the ``distinct goods'' condition, while{, at the same time}, no agent values the bag at least $(k-1)/k$. {But} {T}his means $B$ already contains at least one share of every remaining good in $\mathcal{M}$ and since no good appears more than $|\tilde{N}|$ times in $\mathcal{M}$, it follows for every $i \in \tilde{N}$ that
    \[
        \tilde{v}_i(B) \geq \frac{1}{|\tilde{N}|} \cdot \tilde{v}_i(\mathcal{M}) \geq 1 > \frac{k-1}{k}.
    \]
    This contradicts the assumption of a dead-end state. 
\end{proof}}

{The second issue is handled directly by Lemma~\ref{lem:invariant}. The lemma implies for the last agent $i$ and the remaining shares $\mathcal{M}$ of unique goods,  $\tilde{v}_i(\mathcal{M}) \ge 1$, which means by receiving all of these shares, his utility is at least his MMS value in the instance with agents and goods from the start of Phase 2. By monotonicity of MMS (Appendix~\ref{app:bag-filling}), this completes the argument and yields the following result.}

{The next result follows directly from Propositions~\ref{prop:correctness} and~\ref{prop:termination}, as well as from the monotonicity of MMS, scale invariance, and observation, that the algorithm runs in polynomial time.}

\begin{theorem}\label{prop:Bag-Filling}
    Algorithm \ref{alg:modified-bag-filling} computes in polynomial time the exact MMS $k$-sharing allocation, if $(1-C)(k-1) \geq 1$, or $(1-C)(k-1)$-MMS $k$-sharing allocation otherwise.
\end{theorem}
{\begin{proof}
    All players addressed in Phase 1 receive at least their MMS (guaranteed by monotonicity; Appendix~\ref{app:bag-filling}).
    
    Phase 2 operates on a normalized instance (Appendix~\ref{app:bag-filling}, scale invariance) where $\MMS_i^\text{norm} \le 1$ for all remaining agents. By Proposition~\ref{prop:termination}, no dead-end occurs, therefore each agent except the last receives a bag $B$ satisfying $u_i^\text{norm}(B) \ge (1-C)(k-1)$ (Equation~\eqref{eq:what-player-receives}), and the last agent receives all remaining goods, ensuring utility at least 1.
    
    Thus, the algorithm guarantees an exact MMS for all agents when $(1-C)(k-1) \ge 1$, and $(1-C)(k-1)$-MMS otherwise. Clearly, the algorithm runs in polynomial time in the number of agents and goods.
\end{proof}}

We note that for equal-share cost, where $C = (k-1)/k$, the algorithm produces $\frac{k-1}{k}$-MMS using this algorithm. However, when $C$ is a fixed constant, exact MMS can always be achieved for sufficiently large $k$—specifically, whenever $k \geq 1 + \frac{1}{1 - C}$. Equivalently, for a fixed $k$, exact MMS is attainable as long as $C \leq \frac{k - 2}{k - 1}$. Table~\ref{tab:approx-ratio} illustrates this trade-off between the maximal cost of sharing and sharing degree.

\begin{table}[t]
    \caption{MMS Approximation factor obtained using Algorithm~\ref{alg:modified-bag-filling} for various values of $C$ and $k$}
    \label{tab:approx-ratio}
    \centering
\begin{tabular}{c|ccccccccc}
    $k \backslash C$ & 0.0 & 0.1 & 0.2 & {0.3} & 0.5 & {0.7} & 0.8 & 0.9 & 0.99 \\
    \hline
    2  & \cellcolorval{1.0} & \cellcolorval{0.9} & \cellcolorval{0.8} & \cellcolorval{0.7} &
           \cellcolorval{0.5} & \cellcolorval{0.3} & \cellcolorval{0.2} & \cellcolorval{0.1} & \cellcolorval{0.01} \\
    3  & \cellcolorval{1.0} & \cellcolorval{1.0} & \cellcolorval{1.0} & \cellcolorval{1.0} &
           \cellcolorval{1.0} & \cellcolorval{0.6} & \cellcolorval{0.4} & \cellcolorval{0.2} & \cellcolorval{0.02} \\
    4  & \cellcolorval{1.0} & \cellcolorval{1.0} & \cellcolorval{1.0} & \cellcolorval{1.0} &
           \cellcolorval{1.0} & \cellcolorval{0.9} & \cellcolorval{0.6} & \cellcolorval{0.3} & \cellcolorval{0.03} \\
    5  & \cellcolorval{1.0} & \cellcolorval{1.0} & \cellcolorval{1.0} & \cellcolorval{1.0} &
           \cellcolorval{1.0} & \cellcolorval{1.0} & \cellcolorval{0.8} & \cellcolorval{0.4} & \cellcolorval{0.04} \\
    6  & \cellcolorval{1.0} & \cellcolorval{1.0} & \cellcolorval{1.0} & \cellcolorval{1.0} &
           \cellcolorval{1.0} & \cellcolorval{1.0} & \cellcolorval{1.0} & \cellcolorval{0.5} & \cellcolorval{0.05} \\
    8  & \cellcolorval{1.0} & \cellcolorval{1.0} & \cellcolorval{1.0} & \cellcolorval{1.0} &
           \cellcolorval{1.0} & \cellcolorval{1.0} & \cellcolorval{1.0} & \cellcolorval{0.7} & \cellcolorval{0.07} \\
    10 & \cellcolorval{1.0} & \cellcolorval{1.0} & \cellcolorval{1.0} & \cellcolorval{1.0} &
           \cellcolorval{1.0} & \cellcolorval{1.0} & \cellcolorval{1.0} & \cellcolorval{0.9} & \cellcolorval{0.09} \\
    15 & \cellcolorval{1.0} & \cellcolorval{1.0} & \cellcolorval{1.0} & \cellcolorval{1.0} &
           \cellcolorval{1.0} & \cellcolorval{1.0} & \cellcolorval{1.0} & \cellcolorval{1.0} & \cellcolorval{0.14} \\
    20 & \cellcolorval{1.0} & \cellcolorval{1.0} & \cellcolorval{1.0} & \cellcolorval{1.0} &
           \cellcolorval{1.0} & \cellcolorval{1.0} & \cellcolorval{1.0} & \cellcolorval{1.0} & \cellcolorval{0.19} \\
    25 & \cellcolorval{1.0} & \cellcolorval{1.0} & \cellcolorval{1.0} & \cellcolorval{1.0} &
           \cellcolorval{1.0} & \cellcolorval{1.0} & \cellcolorval{1.0} & \cellcolorval{1.0} & \cellcolorval{0.24} \\
\end{tabular}
\end{table}
\section{Defining Stronger MMS for k-Sharing}\label{sec:smms}

A natural generalization of the maximin share (MMS) to the \( k \)-sharing setting is to guarantee, for each agent \( i \), the maximum over all \( k \)-sharing allocations of the minimum utility they could receive under some assignment of bundles. 
Formally, we introduce \emph{Sharing Maximin Share (SMMS)}.  

\begin{definition}[SMMS]\label{def:smms}
    A $k$-sharing allocation $A$ is said to satisfy \emph{$k$-sharing maximin share (SMMS)} if for each agent $i \in N$: $u_i(A) \geq \SMMS$, where
    \begin{equation}\label{eq:SMMS}
        \SMMS := \max_{B \in \mathcal{A}_k^n(M)} \min_{j \in [n]} u_i(B_{i \leftrightarrow j}),
    \end{equation}
    $\mathcal{A}_k^n(M)$ is the set of all possible $k$-sharing allocations, and $B_{i \leftrightarrow j}$ denotes the modified allocation obtained from $B$ by swapping the bundles assigned to agents $i$ and $j$. {Note that although in this work we assume the sharing-cost and utility of any good depends only on the number of agents sharing the good, the representation of modified allocation as $B_{i \leftrightarrow j}$ allows the sharing-cost and utility to be modeled as functions of the specific agents involved, rather than merely the size of the sharing group.}
\end{definition}

In Definition~\ref{def:smms}, $\min_{j \in [n]} u_i(B_{i \leftrightarrow j})$, captures the worst-case bundle in $B$ for agent~$i${, taking into account that in the $k$-sharing model the utility depends on the entire allocation, not just the agent’s own bundle. The utility $u_i(B_{i \leftrightarrow j})$ can be expressed as}

\[
u_i(B_{i \leftrightarrow j})
= \sum_{g \in B_j} \big[1 - c_g(|N_g(B)|)\big] \cdot v_i(g).
\]

First, we investigate {the existence of} SMMS assuming some restricted cases: (1)~only two agents and (2)~identical valuations. 

\begin{proposition}\label{prop:smms-for-two-agents}
    For any $k$-sharing instance $(N,M,k,v,c)$ with {\textit{generous} cost-sharing model and}  $|N|=k=2$, an SMMS allocation always exists.
\end{proposition}
\begin{proposition}\label{prop:smms-for-identical-agents}
    For any $k$-sharing instance $(N,M,k,v,c)$ with {goods-based cost-sharing instance and} identical valuations $v_i(g)=v_j(g)$ for $i\neq j$ and for all $g\in M$, an SMMS allocation always exists.
\end{proposition}
The proof of Propositions~\ref{prop:smms-for-two-agents} and \ref{prop:smms-for-identical-agents} are provided in Appendices~\ref{app:proof-prop3} and~\ref{app:proof-prop4}, respectively.

For three agents, computing the SMMS becomes significantly more challenging because full-sharing (each good shared by exactly $k$ agents) may not be the best option.

\begin{example}[Full-sharing is not optimal]
Consider an equal-share 2-sharing instance with three agents and two goods, $g_1$ and $g_2$. Agents have identical valuations: $v(g_1) = 1$ and $v(g_2) = 2$. An optimal allocation assigns $g_1$ to one agent, while $g_2$ is shared between the two remaining agents.
\end{example}

In contrast to the equal-share model, the cost-free model always favors full sharing as the optimal choice. We note that the SMMS value of an agent under the generous cost-sharing model is within a factor of $(1 - C)$ of their value in the cost-free model, where $C$ is the maximal cost per good. The details of this comparison are provided in Appendix~\ref{app:comparing-smms}.


Surprisingly, we found that the existing MMS counterexamples~\cite{Kurokawa2018,Feige2021}, in fact, admit SMMS allocations, which we explore next. 

\subsection{Incompatibility of MMS and SMMS}
{In this section, we show that SMMS allocations can exist even in instances where MMS allocations do not. Conversely, we also provide a counterexample demonstrating that an SMMS allocation may fail to exist even when an MMS allocation does exist.}
\citet{Feige2021} considers \( n = 3 \) agents and \( m = 9 \) goods in which no MMS allocation exists under 1-sharing.
{The instance is defined by agent-specific valuations over the goods, represented as \(3 \times 3\) matrices:
\[
\scalebox{0.85}{$
V_1 = \begin{bmatrix}
1 & 16 & 23 \\
26 & 4 & 10 \\
12 & 19 & 9
\end{bmatrix},
V_2 = \begin{bmatrix}
1 & 16 & 22 \\
26 & 4 & 9 \\
13 & 20 & 9
\end{bmatrix},
V_3 = \begin{bmatrix}
1 & 15 & 23 \\
25 & 4 & 10 \\
13 & 20 & 9
\end{bmatrix}.
$}
\]
The goods are indexed from 1 to 9 in row-major order. Although no MMS allocation exists under $1$-sharing, we observe that SMMS allocations do exist for both the \emph{equal-share} and \emph{cost-free} cost-sharing models even under $2$-sharing.
\paragraph{Equal-Share SMMS Allocation}
In the equal-share model, the sum of utilities of all bundles, given an allocation $A$ is equal to $v_i(M)$ for every agent $i$. As $v_i(M)=120$, $\SMMS \leq 40$. At the same time, we have the following allocation with the corresponding utility, proving $\SMMS = 40$ for every agent $i$:
\begin{itemize}
    \item Agent 1 receives goods \( \{1,2,3\} \), utility: \textbf{40.00},
    \item Agent 2 receives goods \( \{4,5,8\} \), utility: \textbf{40.00},
    \item Agent 3 receives goods \( \{6,7,8,9\} \), utility: \textbf{42.00},
\end{itemize}
This is thus SMMS allocation. Notice that it was enough to share only one of the goods to achieve SMMS.
\paragraph{Cost-Free SMMS Allocation}
As the utility of the least bundle in equal-share is bounded by 40, in cost-free sharing, the value of the least bundles is valued at most twice this price, i.e., $\SMMS \leq 80$ for every $i \in N$. The following allocation:
\begin{itemize}
    \item Agent 1 receives goods \( \{1,2,3,4,5,6\} \), utility: \textbf{80.00},
    \item Agent 2 receives goods \( \{2,3,7,8,9\} \), utility: \textbf{80.00},
    \item Agent 3 receives goods \( \{4,5,6,7,8,9\} \), utility: \textbf{81.00},
\end{itemize}
is thus an SMMS allocation.\\
\indent A similar behaviour can be observed for the example from~\citet{Kurokawa2018}, where they constructed an instance with $n=3$ and $m=12$ goods, in which no 1-sharing MMS allocation exists. The valuation of agent $i \in N$ for a good $(k,\ell) \in M$ is given by
\[
v_i(k,\ell) = 10^6 \cdot S_{k,\ell} + 10^3 \cdot T_{k,\ell} + E^i_{k,\ell},
\]
where
\[
S = \begin{bmatrix}
1 & 1 & 1 & 1 \\
1 & 1 & 1 & 1 \\
1 & 1 & 1 & 1
\end{bmatrix}, \quad
T = \begin{bmatrix}
17 & 25 & 12 & 1 \\
2 & 22 & 3 & 28 \\
11 & 0 & 21 & 23
\end{bmatrix}
\]
and
\[
E^1 = \left[\begin{smallmatrix}
3 & -1 & -1 & -1 \\
0 & 0 & 0 & 0 \\
0 & 0 & 0 & 0
\end{smallmatrix}\right],
E^2 = \left[\begin{smallmatrix}
3 & -1 & 0 & 0 \\
-1 & 0 & 0 & 0 \\
-1 & 0 & 0 & 0
\end{smallmatrix}\right],
E^3 = \left[\begin{smallmatrix}
3 & 0 & -1 & 0 \\
0 & 0 & -1 & 0 \\
0 & 0 & 0 & -1
\end{smallmatrix}\right].
\]
\paragraph{Equal-Share SMMS Allocation} For each agent, $v_i(M) = 12,165,000$ for every $i \in N$, thus $\SMMS \leq \frac{v_i(M)}{3} = 4,055,000$. The following allocation:
\begin{itemize}
    \item Agent~1 - \( \{1, 2, 3, 5, 10\} \), utility: \textbf{4,055,001},
    \item Agent~2 - \( \{4, 7, 8, 12\} \), utility: \textbf{4,055,000},
    \item Agent~3 - \( \{5, 6, 9, 10, 11\} \), utility: \textbf{4,055,000}.
\end{itemize}
is thus an SMMS allocation.
\paragraph{Cost-Free SMMS Allocation} Once again, assuming the SMMS value in the cost-free model is at most twice the SMMS value in equal share, we get $\SMMS \leq  8,110,000$. The following allocation is thus an SMMS allocation:
\begin{itemize}
    \item Agent 1 - \( \{1, 2, 3, 5, 7, 8, 10, 12\} \), utility: \textbf{8,110,001},
    \item Agent 2 - \( \{1, 3, 4, 6, 7, 9, 11, 12\} \), utility: \textbf{8,110,001},
    \item Agent 3 - \( \{2, 4, 5, 6, 8, 9, 10, 11\} \), utility: \textbf{8,110,000}.
\end{itemize}
}

{However, under 2-sharing, the following allocation ensures SMMS under cost-free model:}
{Similarly, in the example by~\cite{Kurokawa2018}, the following 2-sharing allocation also gives each agent their SMMS value under cost-free model:}
{Justification for these SMMS allocations and similar results under equal-share cost are deferred to Appendix~\ref{app:smms-superior}.}

Given these examples, it may seem that achieving SMMS is easier than achieving MMS; however, we provide a counterexample for SMMS by using the structure of the instance proposed by~\citet{Kurokawa2018}{,} where SMMS does not exist, even though MMS does.


\begin{theorem}\label{prop:smms-dont-exist}
    There exists an instance of cost-free $2$-sharing problem with $n=3$ and $m=12$, which does not admit an SMMS allocation, but admits MMS allocation.
\end{theorem}

\begin{proof}
Let $M = \{ (k,\ell) \mid k \in [3], \ell \in [4]\}$. The valuation of agent $i \in N$ for a good $(k,\ell) \in M$ is given by
\[
v_i(k,\ell) = 10^7 \cdot S_{k,\ell} + 10^3 \cdot T_{k,\ell} + E^i_{k,\ell},
\]
\[
S = \begin{bmatrix}
1 & 1 & 1 & 1 \\
1 & 1 & 1 & 1 \\
1 & 1 & 1 & 1
\end{bmatrix}, \quad
T = \begin{bmatrix}
17 & 25 & 12 & 1 \\
2 & 22 & 3 & 28 \\
11 & 0 & 21 & 23
\end{bmatrix},
\]
\[
E^1 = \left[\begin{smallmatrix}
-2 & 1 & 0 & 1 \\
0 & 0 & 0 & 0 \\
-1 & 0 & 1 & 0
\end{smallmatrix}\right],
E^2 = \left[\begin{smallmatrix}
0 & 0 & -2 & 0 \\
0 & 1 & 0 & 0 \\
0 & 1 & 0 & 0
\end{smallmatrix}\right],
E^3 = \left[\begin{smallmatrix}
-1 & -1 & 0 & 0 \\
0 & 0 & 1 & 0 \\
0 & 1 & 0 & 0
\end{smallmatrix}\right].
\]
The choice of matrix $S$ enforces that for each agent $i$, his \SMMS[i][2][3] is achieved for an allocation which assigns each good to exactly two agents, i.e. each agent receives 8 goods. Under such $2$-sharing allocations, $\SMMS[i][2][3] \sim 8 \cdot 10^7$ for every agent, which cannot be achieved for allocations where any share contains less than 8 goods (in such case, such agent receives strictly less than $8 \cdot 10^7$).

Matrix $T$ was chosen in such a way that if we label it using three different types of labels (numeric: $1,2,3$, Greek: $\alpha,\beta,\gamma$, and symbols: $+,-,*$), entries corresponding to one specific label sum to exactly 55:
\[
T = \begin{bmatrix}
{}^{\alpha}17^{1}_{+} & {}^{\alpha}25^{1}_{-} & {}^{\beta}12^{1}_{+} & {}^{\gamma}1^{1}_{*} \\
{}^{\alpha}2^{2}_{-} & {}^{\beta}22^{2}_{*} & {}^{\gamma}3^{2}_{+} & {}^{\gamma}28^{2}_{-} \\
{}^{\alpha}11^{3}_{+} & {}^{\beta}0^{3}_{-} & {}^{\beta}21^{3}_{*} & {}^{\gamma}23^{3}_{+}
\end{bmatrix}.
\]
Since the total sum of entries in matrix \( T \) is \( 3 \times 55 = 165 \), any other partitioning of the goods into three bundles of four must include one whose sum in \( T \) is strictly less than 55.

Now, consider an allocation that divides $M$ into three bundles, each containing 8 goods. Each bundle is constructed by selecting goods associated with two labels of the same type—for instance, one bundle includes goods labeled with \( \alpha \) and \( \beta \), another with \( \beta \), \( \gamma \), and the third with \( \alpha \), \( \gamma \). In this allocation, the sum of entries in \( T \) for each bundle is exactly 110. By the argument above, any allocation that does not follow this label-based structure yields at least one bundle for which the sum of entries in \( T \) is strictly less than 110.

This means, if there is an SMMS allocation, it must be one of those described above; assigning pairs of labels to agents
\begin{enumerate}
    \item according to 1,2,3,
    \item according to $\alpha$, $\beta$, $\gamma$,
    \item according to $+$, $-$, $*$.
\end{enumerate}
    Each of the three ways guarantees \SMMS[i][2][3] to one of the agents $i$, which is the same for all of the agents and equal to $8,110,000$. However, none of the three ways assign \SMMS[i][2][3] to all of the agents at the same time. This is enforced by the structure of matrices $E^1$, $E^2$, and $E^3$.

    In division according to 1,2,3, both agent 2 and agent 3 need to be assigned elements with labels 2 and 3, otherwise their value is $8,109,999$. Similarly, in division according to $\alpha, \beta, \gamma$, both agents 1 and 3 need to be assigned elements with $\beta$ and $\gamma$, or, once again, their value is less than their SMMS value. Finally, in division according to $+,-,*$, both agents 1 and 2 need to be assigned elements with $-$ and $*$.

    Finally, assigning each agent their corresponding row (i.e., the first row to the first agent and so on) yields a valid $1$-sharing MMS allocation.
\end{proof}

Despite the lack of general existence, we show that an $\alpha(1-C)$-SMMS allocation of a goods-based cost-sharing model with $C$ being the maximal cost can be computed in polynomial time, via a reduction to $\alpha$-CMMS. 

\section{Approximate SMMS}\label{sec:smms-equiv-cmms}
We provide a connection of SMMS with the cardinality-constrained maximin share (CMMS). An instance of the \emph{fair allocation problem under cardinality constraints}~\cite{Hummel2022,Biswas2018} is denoted as $(N,M,v,b)$ where $N$ is the set of agents, $M = \{C_1,\dots,C_\ell\}$ is a set of goods partitioned into sets $C_i$ according to $\ell$ types and $b = (b_1,\dots,b_\ell)$ is the budget profile of types. In this setting, only a subset of allocations $\mathcal{F} \subseteq A^n_1(M)$ is feasible; under feasible allocation, no player can receive more than $b_i$ goods of type $i$. Formally, $A \in \mathcal{F}$ satisfies for every $i \in [n]$, $j \in [\ell]$ that $|A_i \cap C_j| \leq  b_j$. Here, the \emph{cardinality-constrained maximin share} (CMMS) value is
\begin{equation}
\text{CMMS}^n_i(M) := \max_{A \in \mathcal{F}} \min_{j \in [n]} v_i(A_j).
\end{equation}
In any goods-based cost-sharing setting with $c_{g}(|N_g(A)|)\in[0,1]$, fully-shared SMMS allocations correspond to allocations of CMMS of a special case of the fair allocation model under cardinality constraints, as stated in Proposition~\ref{prop:smms-equiv-cmms} (see Appendix~\ref{app:proof-prop5} for more details and proof).

\begin{proposition}\label{prop:smms-equiv-cmms}
For every instance $I=(N, M, u)$ of the goods-based $k$-sharing problem, there exists an instance $\tilde{I} = (N, \tilde{M}, \tilde{v}, b)$ of a fair allocation under cardinality constraints such that fully-shared allocations $A \in \mathcal{A}_k^n$ correspond bijectively to feasible allocations $\tilde{A} \in \mathcal{F}$, and the bijection preserves agents' utilities: $\forall i \in N, u_i(A) = \tilde{v}_i(\tilde{A}_i)$.
\end{proposition}

If we define the \emph{full-sharing maximin share value} \fullSMMSc[][i] by restricting the SMMS definition to full-$k$-sharing allocations, then it follows from Proposition~\ref{prop:smms-equiv-cmms} that \( \operatorname{CMMS}_i = \fullSMMSc[] \) for every \( i \in N \) and also any \( \alpha \)-CMMS allocation \( \tilde{A} \) corresponds to an \( \alpha \)-\fullSMMSc[][] allocation \( A \). As \fullSMMSc[][i] is defined over less allocations than \SMMSc[][i], we have $\SMMSc[] \geq \fullSMMSc[]$. More importantly, we can bound $\SMMSc[]$ and $\MMSi$ by \fullSMMSc[][i] from above, allowing us to construct approximations of these by $\alpha$-$\operatorname{CMMS}$ allocations. 
\begin{lemma}\label{lem:fullSMMS-bounding}
    For every goods-based $k$-sharing model with maximal possible sharing cost $C$, it holds
    \begin{enumerate}
        \item $\fullSMMSc[] \geq (1-C) \cdot \SMMSc[]$,
        \item $\fullSMMSc[] \geq k \cdot (1-C) \cdot \MMSi$.
    \end{enumerate}
\end{lemma}
\begin{proof}
    (1.) Let a \( k \)-sharing allocation \( A \) be an \SMMSc[] maximizer , and let \( B \) be a fully-shared allocation obtained by taking \( A \) and assigning the remaining shares of each good to arbitrary agents. Consequently, \( A_j \subseteq B_j \) for every \( j \in N \).
    Since $A$ is \SMMSc[] maximizer, \SMMSc[] is equal to
    \begin{align}\label{eq:smms-upper-bound}
    \min_{j \in N} u_i(A_{i \leftrightarrow j}) 
    &= \min_{j \in N} \sum_{g \in A_j} \left(1 - c_g\left(|N_g(A_{{i \leftrightarrow j}})|\right)\right) \cdot v_i(g) \notag \\
    &\leq \min_{j \in N} \sum_{g \in A_j} v_i(g)
    \end{align}
    where the last inequality follows from the fact that removing costs increases utility. Further, since $B$ is a fully-shared allocation, we have that \fullSMMSc[] is at least
    \begin{align}\label{eq:fullSMMS-lower-bound}
        \min_{j \in N} u_i(B_{i \leftrightarrow j}) 
        &= \min_{j \in N} \sum_{g \in B_j} (1 - {c_g(}|N_g(B_{{i \leftrightarrow j}})|{)}) \cdot v_i(g) \notag \\
        &\geq (1-C) \cdot \min_{j \in N} \sum_{g \in B_j} v_i(g)
    \end{align}
    where the final inequality uses that $C$ is the maximum cost. Finally, since $A_j \subseteq B_j$, we have
    \begin{equation*}
        (1-C) \cdot \min_{j \in N} \sum_{g \in B_j} v_i(g) \geq (1-C) \cdot \min_{j \in N} \sum_{g \in A_j} v_i(g)
    \end{equation*}
   which in combination with~\eqref{eq:smms-upper-bound} and ~\eqref{eq:fullSMMS-lower-bound} concludes the proof.

    (2.) Let a 1-sharing allocation $A \in \mathcal{A}_1^n$ be an $\MMSi$ maximizer. We construct a different fully-shared allocation \( B \) by giving each agent their own bundle from \( A \) and the next \( k - 1 \) consecutive bundles, looping around to the beginning if we run out. Without the costs, each agent $i$ values each bundle in $B$ more than $k$-times the least valuable bundle in $A$, i.e., $\forall j \in N$, we have
    \begin{equation}\label{eq:new-bundle-k-better}
        v_i(B_j) \geq k \cdot \min_{\ell \in N} v_i(A_\ell) \geq k \cdot \MMSi.
    \end{equation}
   Note that~\eqref{eq:fullSMMS-lower-bound} still holds for this allocation \( B \), since the inequality relied only on \( B \) being full-sharing. Combining~\eqref{eq:fullSMMS-lower-bound} with~\eqref{eq:new-bundle-k-better} therefore gives us: 

    \begin{equation*}
        \fullSMMSc[] \geq (1-C) \cdot \min_{j \in N}  v_i(B_j) \geq k {\cdot (1-C)} \cdot \MMSi.
    \end{equation*}
\end{proof}

By combining the equivalence between $\operatorname{CMMS}$ and $\fullSMMSc[][]$ with the upper bounds 
from Lemma~\ref{lem:fullSMMS-bounding}, we obtain approximation results for both MMS and SMMS.

\begin{proposition}\label{prop:approx-MMS-SMMS}
    Let $\alpha \in (0,1]$, and suppose that for every fair division instance with cardinality constraints, there exists an $\alpha$-$\operatorname{CMMS}$ allocation. Then, for any goods-based $k$-sharing model with maximum cost $C$, there exist:
    \begin{itemize}
        \item an $\alpha \cdot (1 - C)$-$\operatorname{SMMS}$ allocation, and
        \item an $\alpha \cdot k \cdot (1 - C)$-$\operatorname{MMS}$ allocation.
    \end{itemize}
\end{proposition}


The best known \( \alpha = 1/2 \)~\cite{Hummel2022} gives a \( \frac{1}{2} \)-SMMS guarantee in the cost-free model and \( \frac{1}{2k} \)-SMMS under equal-share. Also, while a larger $k$ leads to a worse approximation factor, the SMMS value may increase due to the greater flexibility in sharing.

Regarding MMS approximations, we note that \( \alpha = 1/2 \) does not improve upon the bag-filling algorithm (Proposition~\ref{prop:Bag-Filling}), which guarantees \( (k - 1)(1 - C) \)-MMS. However, for \( \alpha > \frac{k - 1}{k} \), the approach via \( \alpha \)-CMMS yields better guarantees—meaning that improvements occur even for \( \alpha = 1/2 + \varepsilon \) with arbitrarily small \( \varepsilon > 0 \).

\section{Conclusion}

In this paper, we studied how allowing limited sharing of indivisible goods affects fairness, specifically focusing on Maximin Share (MMS) fairness. By permitting goods to be shared among up to $k$ agents, we found that previously unattainable fairness guarantees become achievable. We introduced Sharing Maximin Share (SMMS) fairness and for both notions explored when fair solutions exist, revealing important trade-offs between fairness, sharing costs, and the maximum number of agents allowed to share each good.

{Our analysis highlighted the challenges arising when agents have different opinions on sharing—some may prefer to share goods, while others may not—particularly under equal-share. These differences underline the difficulty of designing fair allocation methods acceptable to all agents.}

Our findings show that allowing {bounded cost-sensitive} sharing greatly improves the practicality of fair allocations, {provided that the costs of sharing are not excessive,} making {our proposed mechanisms} suitable for scenarios where sharing is possible. {Table~\ref{tab:approx-ratio} could be viewed as a baseline for future work refining the boundary where sharing can or cannot lead to improved MMS guarantees depending on the parameters $k$ and $C$}. A key open question is determining the lower bound on the number of agents $k$ needed to achieve fairness under \textit{generous cost}{\textit{-sharing} model}.

{The theoretical results established in this paper open up new research directions for fair resource allocation when a limited sharing of resources are allowed in multi-agent environments.} {It} opens the door to exploring various cost-sharing models that reflect structured, real-world constraints. One promising direction is to consider models where the cost of sharing is determined not by the goods themselves, but by the group {of other agents} they are sharing with---for example, agents may differ in how effectively they can make use of shared resources, with some incurring lower costs when participating in sharing, and costs being divided equally when similarly capable agents share the same good.



\begin{acks}
This work was carried out while H.S. and M.C. were participants in the 2025 DIMACS REU program at Rutgers University, supported by the NSF grant CCF-2447342. H.S. and M.C. were partly supported by the RSJ Foundation and by the Department of Applied Mathematics and the Computer Science Institute of Charles University. H.S. was supported by the Horizon Europe Programme under Grant Agreement No. 101183743 (AGATE).
\end{acks}



\balance
\bibliographystyle{ACM-Reference-Format} 
\bibliography{subs}


\clearpage
\appendix

\section{Details of the Algorithm~\ref{alg:modified-bag-filling}}\label{app:bag-filling}

At the Phase 1 of Algorithm~\ref{alg:modified-bag-filling}, large goods are given to respective agents with an assumption, that this process will reduce the problem of MMS allocation between a smaller number of agents with smaller number of small goods. This reduction is justified by monotonicity of the MMS, a result which was already proved for example by~\citet{Feige2021}. For the sake of completeness, we provide the proof.

\begin{proposition}[Monotonicity of MMS]\label{prop:monotonicity-of-MMS}
    Let $I = (N, M, k, v, c)$ be an instance of the fair division problem. For every $i \in N$ and every $g \in M$, it holds
    \[\MMS[i][n-1][M \setminus \{g\}]\geq \MMS[i][n][M].\]
    
\end{proposition}
\begin{proof}
    Let $A = (A_1,\dots,A_n)$ be an $\MMSi$ maximizer and assume without loss of generality that $g \in A_n$, $n\neq i$. Then for allocation of $n-1$ agents, $B = (A_1,\dots,A_{n-2},A_{n-1} \cup A_n \setminus \{g\})$, the minimal bundle is at least the minimal bundle of $A$, proving the monotonicity.
\end{proof}

At the beginning of Phase 2 of Algorithm~\ref{alg:modified-bag-filling}, $|\tilde{N}|\geq 1$ and, unless $\MMS[i]=0$ in the original instance, the value of the remaining goods is $v_i(\tilde{M})>0$ for all $i\in\tilde{N}$. Therefore, we scale valuations of remaining agents by factor $\frac{|\tilde{N}|}{v_i(\tilde{M})}>0$, so that the MMS value of every remaining agent is bounded from above by $1$. The following proposition shows that this does not ruin the approximation guarantee. Corresponding result for classical setting was used and proved by~\citet{garg2019approximating}, our proof is a natural generalization of their proof for $k$-sharing setting.

\begin{proposition}[Scale Invariance]
    Let $A=(A_1,\dots,A_n)$ be an $\alpha$-MMS $k$-sharing allocation for instance $I = (N, M, v,k, c)$. If valuations are scaled by a positive factor $r>0$, i.e., $v^r_i(g) = r \cdot v_i(g)$ for every $i \in N$, allocation $A$ remains $\alpha$-MMS in the scaled instance $I^{r} = (N, M, k, v^r,c)$.
\end{proposition}
\begin{proof}
    Let $M_i$ and $M^{r}_i$ denote the MMS value of agent $i$ in instance $I$ and $I^{r}$ respectively. For any $k$-sharing allocation $A \in \mathcal{A}_k^n$, we have
    \begin{align*}
    u_i^r(A) &= \sum_{g \in A_i} (1 - c_{i,g}(N_g(A))) \cdot v_i^r(g) \\
             &= \sum_{g \in A_i} (1 - c_{i,g}(N_g(A))) \cdot r \cdot v_i(g) \\
             &= r \sum_{g \in A_i} (1 - c_{i,g}(N_g(A))) \cdot v_i(g) \\
             &= r \cdot u_i(A).
    \end{align*}
    Therefore $M^r_i = r \cdot M_i$. Let $A$ be an $\alpha$-MMS $k$-sharing allocation for instance $I$. Since $u^r_i(A) = r \cdot u_i(A) \geq r \cdot \alpha \cdot M_i = \alpha \cdot M^r_i$, $A$ is an $\alpha$-MMS allocation even in the scaled instance $I^r$.
\end{proof}





\section{SMMS with Two Agents}\label{app:proof-prop3}

\begin{repproposition}
    For any generous goods-based $2$-sharing instance $(N,M,k,v,c)$ with $|N|=2$, an SMMS allocation always exists.

    \begin{proof}
    We show that the 2-sharing allocation \( C = (M, M) \), in which all goods are assigned to both agents, is an SMMS allocation. To this end, let \( A = (B_1 \cup S, B_2 \cup S) \) be an SMMS maximizer for agent $i$ where \( B_1 \cap B_2 = \emptyset \), and suppose agent \( i \in N \) receives the smaller share, i.e., \( v_i(B_i) \leq v_i(B_j) \). This means that
    \[
    v_i(B_i) \leq \frac{v_i(B_i) + v_i(B_j)}{2} \leq u_i(D),
    \]
    where in $2$-sharing allocation \( D \), agent \( i \) receives only the set \( B_i \cup B_j \), shared entirely with agent \( j \), and no other goods. The difference between \( u_i(D) \) and \( u_i(C) \) is exactly the utility agent \( i \) receives from the shared goods \( S \), which also accounts for the difference between \( u_i(A) \) and \( v_i(B_i) \). It therefore follows that \( u_i(A) \leq u_i(C) \), and so the fully-shared allocation \( C \) satisfies the SMMS condition.
\end{proof}
\end{repproposition}

\section{SMMS under Identical Utilities}\label{app:proof-prop4}
We initiate this section by proving Proposition~\ref{prop:smms-for-identical-agents}, before we show that in general, identical utilities appear only in goods-based cost-sharing models.
\begin{repproposition3}
    For any goods-based cost-sharing instance $(N,M,k,v,c)$ with identical valuations $v_i(g)=v_j(g)$ for $i\neq j$ and for all $g\in M$, an SMMS allocation always exists.
\end{repproposition3}

First, we can observe that if all agents have identical utilities, an SMMS allocation always exists for any instance $(N,M,k,v,c)$. The proof is straightforward, relying on the fact that any allocation that is an SMMS maximizer for one agent is an SMMS maximizer for all other agents, therefore it is an SMMS allocation.

\begin{observation}\label{obs:dentical-utilities}
    For any instance $(N,M,k,v,c)$ with identical utilities, i.e., $u_i(A)=u_j(A_{i \leftrightarrow j})$ for all $A  \in \mathcal{A}_k^n$ and all $i, j \in N$, an SMMS allocation always exists.
\end{observation}
    
We further establish in Proposition~\ref{prop:identical-utilities} that in the $k$-sharing setting identical utilities occur precisely when agents have identical valuations and the cost structure is goods-based.

\begin{proposition}\label{prop:identical-utilities}
In a $k$-sharing instance $(N, M, k,v,c)$, all agents have identical utility functions, i.e.,
\[u_i(A) = u_j(A_{i \leftrightarrow j})\text{ for all }i,j \in N, A  \in \mathcal{A}_k^n\]
 if and only if all agents have identical valuations and the costs are goods-based.
\end{proposition}

\begin{proof}
($\Leftarrow$) If all agents have identical valuations and the costs are goods-based, then utility depends only on the allocation structure, not the agent's identity. Therefore, all agents have identical utility functions.

($\Rightarrow$) Assume that all agents have identical utility functions. Fix a good \( g \in M \), and consider any 1-sharing allocation \( A \) in which agent \( i \) is allocated only good \( g \). Then for any \( i, j \in N \), we have
\[
    v_i(g) = v_i(A_i) = u_i(A) = u_j(A_{i \leftrightarrow j}) = v_j(A_i) = v_j(g),
\]
so all agents have identical valuations.

Now consider two groups \( S, S' \subseteq N \) of the same size, along with agents \( i \in S \) and \( j \in S' \). To show that \( c_{i,g}(S) = c_{j,g}(S') \), consider an arbitrary set \( T \subseteq N \) of size at most \( k \), and an allocation \( A \) in which good \( g \) is shared by all agents in \( T \), while the remaining goods are allocated arbitrarily among the agents in \( N \setminus T \) {, while preserving $k$-sharing
}. Further, fix $x \in T$ and $y \in N \setminus x$. We make two observations based on whether \( y \in T \), the relation \( u_x(A) = u_y(A_{x \leftrightarrow y}) \), and the fact that both agents receive only good \( g \) in the respective allocations:

\begin{enumerate}
    \item If \( y \in T \), then \[(1 - c_{x,g}(T)) \cdot v(g) = (1 - c_{y,g}(T)) \cdot v(g)\]
    \[\implies c_{x,g}(T) = c_{y,g}(T).\]
    \item If \( y \notin T \), then 
    \[(1 - c_{x,g}(T)) \cdot v(g) = (1 - c_{y,g}(T \setminus x \cup y)) \cdot v(g)\] \[\implies c_{x,g}(T) = c_{y,g}(T \setminus x \cup y).\]
\end{enumerate}
The first observation means that when a group of agents shares a good among themselves, each agent incurs the same cost for the sharing. The second observation, in combination with the first, implies that if two equally-sized groups differ by exactly one agent (i.e., one agent from one group can be swapped for one from the other), then the cost of sharing the good in the first group is the same for each of its agent as the cost of sharing the good in the other group for each of its agent. Extending on this argument, we can gradually swap agents \( x \in S \setminus S' \) with agents \( y \in S' \setminus S \) to transform \( S \) into \( S' \), preserving the cost at each step. This shows that agents in \( S \) incur the same cost for sharing good \( g \) as those in \( S' \), which concludes the proof.

\end{proof}

Proposition~\ref{prop:smms-for-identical-agents} simply follows from the Observation~\ref{obs:dentical-utilities} and Proposition~\ref{prop:identical-utilities}.

\section{Comparing SMMS Across Cost Models}\label{app:comparing-smms}

Before turning to structural relationships between fairness notions, we briefly compare the SMMS values obtained under different cost-sharing models. In particular, we consider how the cost-free and generous models relate.

\begin{lemma}
    Let \SMMSc[cf] be the SMMS value of agent $i$ in the cost-free model and \SMMSc the SMMS value of agent $i$ in a generous cost-sharing model with maximal cost for sharing a good being $C$. Then it holds
    \begin{equation*}
        \SMMSc[cf] \geq \SMMSc \geq (1-C)\SMMSc[cf].
    \end{equation*}
\end{lemma}
\begin{proof}
    For the first inequality, let \(k\)-sharing allocation \( A \) be an \SMMSc maximizer. Since the cost-free model imposes no cost on sharing, the utility of agent~\(i\) from allocation \( A \) is at least as high in the cost-free model as in the model with costs. Therefore, \(\SMMSc[cf] \geq \SMMSc\).

    For the second inequality, let  full-sharing allocation \( A \) be an $\SMMSc[cf]$ maximizer. In the model with sharing costs bounded by \( C \), each good contributes at least a \( (1 - C) \) fraction of its value compared to the cost-free model. Since utility is additive, the total utility of agent~\(i\) in this model is at least \( (1 - C) \SMMSc[cf] \). Hence, \(\SMMSc \geq (1 - C) \SMMSc[cf]\), as claimed.
\end{proof}

\section{Relationship between SMMS and CMMS}\label{app:proof-prop5}

\begin{repproposition2}

For every instance $I=(N, M, u)$ of the goods-based $k$-sharing problem, there exists an instance $\tilde{I} = (N, \tilde{M}, \tilde{v}, b)$ of a fair allocation under cardinality constraints such that fully-shared allocations $A \in \mathcal{A}_k^n$ correspond bijectively to feasible allocations $\tilde{A} \in \mathcal{F}$, and the bijection preserves agents' utilities: $\forall i \in N, u_i(A) = \tilde{v}_i(\tilde{A}_i)$.

\begin{proof}
    Let $(N,M,k,v,c)$ be an instance of a cost-free $k$-sharing problem  with $M = \{g_1, \dots g_m\}$. Let us now define an instance of fair division problem under cardinality constraints $(N,\tilde{M},\tilde{v},b)$ satisfying $\tilde{M} = \{G_1,\dots, G_m\}$ with $G_i = \{g_i^j \mid j \in [k]\}$, the unit budget profile $b = (1,\dots,1)$, and $\tilde{v}_i(g_x^y) = (1-c_g(k))v_i(g_x)$ for every $x \in [m]$, $y \in [k]$, $i\in N$.

    For every allocation $\tilde{A} \in \mathcal{F}$ (feasible according to the cardinality constraints), there is an allocation $A \in \mathcal{A}_k^n(M)$, defined as $A_i = \{g_i \mid g_i^j \in \tilde{A}_i\}$ for every $i \in N$. For both of these allocations, $\tilde{v}_i(\tilde{A}_j) = v_i(A_j)$, as
    \begin{equation*}
        \tilde{v}_i(\tilde{A}_j) = \sum_{g_x^y \in \tilde{A}_j}\tilde{v}_i(g_x^y) = \sum_{g_x \in A_j}(1-c_g(k))v_i(g_x) = u_i(A_j).
    \end{equation*}

    Similarly, for every instance of fair division problem under cardinality constraints $(N,\tilde{M},b,\tilde{v})$ satisfying $\tilde{M} = \{C_1,\dots,C_\ell\}$ with $|C_i| = |C_\ell| = k$, $b = (1,\dots,1)$ and $\tilde{v}_i(g_1) = \tilde{v}_i(g_2)$ for every $i \in N$, $g_1,g_2 \in C_j$, $j \in [\ell]$, one can construct an instance of the goods-based $k$-sharing problem $(N,M,k,v,c)$ where goods in $M$ are the sets $C_i$ themselves and the utility of $v_i(C_i) = \tilde{v}_i(g)/(1-c_g(k))$ for any $g \in C_i$. For such a pair of instances, any full-$k$-sharing allocations (sharing all goods among exactly $k$ agents) correspond to feasible allocations in $\mathcal{F}$ with $\tilde{v}_i(\tilde{A}_j) = u_i(A_j)$ for every $A$ and $\tilde{A}$.
\end{proof}

\end{repproposition2}


\end{document}